\documentclass[sigconf]{acmart}
\renewcommand\footnotetextcopyrightpermission[1]{} 
\usepackage{microtype,graphicx,amsfonts,amssymb,mathrsfs}

\copyrightyear{2017}
\acmYear{2017}
\setcopyright{acmlicensed}
\acmConference{SIGSPATIAL'17}{November 7--10, 2017}{Los Angeles
Area, CA, USA}\acmPrice{15.00}\acmDOI{10.1145/3139958.3139999}
\acmISBN{978-1-4503-5490-5/17/11}

\begin{document}

\title{Crossing Patterns in Nonplanar Road Networks}

\author{David Eppstein}
\affiliation{%
  \institution{Dept. of Computer Science}
  \streetaddress{Univ. of California, Irvine}
  \city{Irvine} 
  \state{CA} 
  \postcode{92697}
}
\email{eppstein@uci.edu}

\author{Siddharth Gupta}
\affiliation{%
  \institution{Dept. of Computer Science}
  \streetaddress{Univ. of California, Irvine}
  \city{Irvine} 
  \state{CA} 
  \postcode{92697}
}
\email{guptasid@uci.edu}

\begin{abstract}
We define the \emph{crossing graph} of a given embedded graph (such as a road network) to be a graph with a vertex for each edge of the embedding, with two crossing graph vertices adjacent when the corresponding two edges of the embedding cross each other.
In this paper, we study the sparsity properties of crossing graphs of real-world road networks.
We show that, in large road networks (the Urban Road Network Dataset), the crossing graphs have connected components that are primarily trees, and that the  remaining non-tree components are typically sparse (technically, that they have bounded degeneracy). We prove theoretically that when an embedded graph has a sparse crossing graph, it has other desirable properties that lead to fast algorithms for shortest paths and other algorithms important in geographic information systems. Notably, these graphs have \emph{polynomial expansion}, meaning that they and all their subgraphs have small separators.
\end{abstract}



\maketitle


\section{Introduction}
Road networks are often modeled graph-theoretically, by placing a graph vertex at each intersection or terminus of roads, and connecting vertices by edges that represent each segment of road between two vertices. Thus, each vertex is naturally associated with a coordinate on the earth's surface.

Much past work on algorithms for road networks has either assumed that these networks are planar (that is, that no two roads cross without forming an intersection at their crossing point) or has added artificial intersection points to roads that cross without intersection, to force these networks to be planar. Planar graphs have  many convenient  properties, including planar graph duality and planar graph separator theorems~\cite{LipTar-SJAM-79,Chu-PFV-90} that allow natural and important problems on these networks to be solved more quickly. For instance, for planar graphs, it is known how to compute shortest paths in linear time, based on the planar separator theorem~\cite{KleRaoRau-STOC-94}, in contrast to the situation for general graphs where shortest paths are slower by a logarithmic factor.
Unfortunately, as Eppstein et al.~\cite{eppstein2008studying} observed, the available data for real-world road networks shows that these networks are not actually planar: they include many crossings.
This discovery naturally raises the question of how to model nonplanar road networks, in a way that allows efficient algorithms to be based on their properties.

In this context, one would like a model of road networks that is  in some sense near-planar (after all, road networks have few points where roads cross without intersecting, although their number is not zero), that is realistic (accurately modeling real-world road networks), and that is useful (leading to efficient algorithms).

One clear property of road networks is that they are \emph{sparse}: the number of road segments exceeds the number of road intersections by only a small factor, half of the average number of segments that meet at an intersection. Since the vast majority of intersections are the meeting point of three or four road segments, this means that the number of road segments should be between 1.5 and 2 times the number of intersections. Researchers in graph algorithms and graph theory have developed a sophisticated hierarchy of classifications of sparse graph families centered around the intuitive notion of sparseness. Many of these types of sparseness imply general algorithmic meta-theorems about the properties that can be computed efficiently for graphs in the given family.
In particular, many of the known algorithms for planar graphs can be extended to the class of graphs of \emph{polynomial expansion}, a property that was originally defined using graph minor theory but that has a more natural equivalent definition (for classes of graphs closed under taking subgraphs) in terms of the existence of sublinear-size separators~\cite{NevOss-Sparsity-12,DvoNor-SSS-15}. Graphs of polynomial expansion support efficient separator-based divide-and-conquer algorithms, as well as more sophisticated pattern matching algorithms based on their graph minor properties.
We would like to show that road networks, too, have small separators, and therefore that they can support all of these algorithms.

In this paper, we provide a mathematical model of non-planar road networks in terms of the sparseness of their \emph{crossing graphs}, graphs representing pairs of road segments that cross in the network. We analyze the Urban Road Network Data set and show that, indeed, it is a good fit for the model. Additionally, we prove that networks within this model have polynomial expansion, from which it follows that the linear-time planar shortest path algorithm of Klein et al~\cite{KleRaoRau-STOC-94} can be adapted to work on these networks, despite their non-planarities.

\section{Past work}
\subsection{Nonplanar road networks}

The past work by Eppstein et al.~\cite{eppstein2008studying,eppstein2009going,eppstein2009linear} has attempted to model nonplanarities in planar road networks in two different ways.
In \cite{eppstein2008studying} the authors posited that road networks are subgraphs of the intersection graphs of systems of disks (the disks centered at each intersection of roads with radius equal to half the length of the longest segment of roads meeting at that intersection) and that, with a small number of exceptional high-radius disks, these disks have low \emph{ply} (at most a constant number of non-exceptional disks cover any point of the Earth's surface). They performed an empirical analysis of road network data showing that this model fits actual road networks reasonably well, and they used the assumption to develop efficient road network algorithms. Unfortunately, the class of networks defined by this model does not fit well into the theory of  sparse networks, because of its handling of the exceptional disks. Because these disks could be arbitrary, they could in principle heavily overlap each other, producing dense subgraphs that prevent the graphs defined in this way from having polynomial expansion or other good sparseness properties. This misbehavior seems unlikely to happen in actual road networks, but this mismatch between theory (where these graphs can have dense subgraphs) and practice (where dense subgraphs are unlikely) indicates that it should be possible to replace their model of road networks by another model that more accurately matches the sparsity properties of real-world road networks.

Later work of the same authors~\cite{eppstein2009going} attempted to justify the low number of crossings in road networks by showing that randomly chosen lines (modeling, for instance, a highway cutting across an older city grid) typically have a sublinear number of crossings with other roads.
Another paper~\cite{eppstein2009linear} used this observation of few crossings per line as the basis for a very weak assumption about road networks, that the total number of crossings is smaller than the number of intersections by a sufficiently large non-constant factor. This assumption allows some algorithms to be performed efficiently; notably, the crossings themselves can all be found in linear time. However, it is not strong enough to imply the existence of small graph separators for arbitrary subgraphs of road networks, a property that is required by fast separator-based graph algorithms. Additionally, this paper's assumptions about the sparsity of crossings are on dubious ground from an empirical point of view: what reason is there to believe that the ratio of intersections to crossings in large street networks is non-constant rather than a large constant?

\subsection{Nearly-planar graphs}
\label{sec:past-nearly-planar}

The graph theory literature includes many natural generalizations of planar graphs that we may choose from. Among these, the $k$-apex graphs have been defined as the graphs that can be made planar by the removal of $k$ vertices~\cite{Kaw-FOCS-09}; however, in road networks, the number of vertices that would need to be removed would typically be proportional to the number of crossings, a large enough number that it is not reasonable to treat it as a constant. Similarly, the $k$-genus graphs are the graphs that can be embedded without crossings into a surface of genus at most $k$~\cite{DemFomHaj-JACM-05}. A road network can be embedded without crossings on a surface with a handle for each overpass or tunnel, and this surface would have genus proportional to the number of overpasses and tunnels, but again this number would be too large to treat as a constant.
These two graph families form minor-free graph families: families of graphs that, like the planar graphs, have some forbidden graphs that cannot be formed from contractions or deletions of graphs in their family. However, there is no reason to expect any particular forbidden minors in road networks.

\begin{figure}[t]
\centering
\includegraphics[width=0.65\columnwidth]{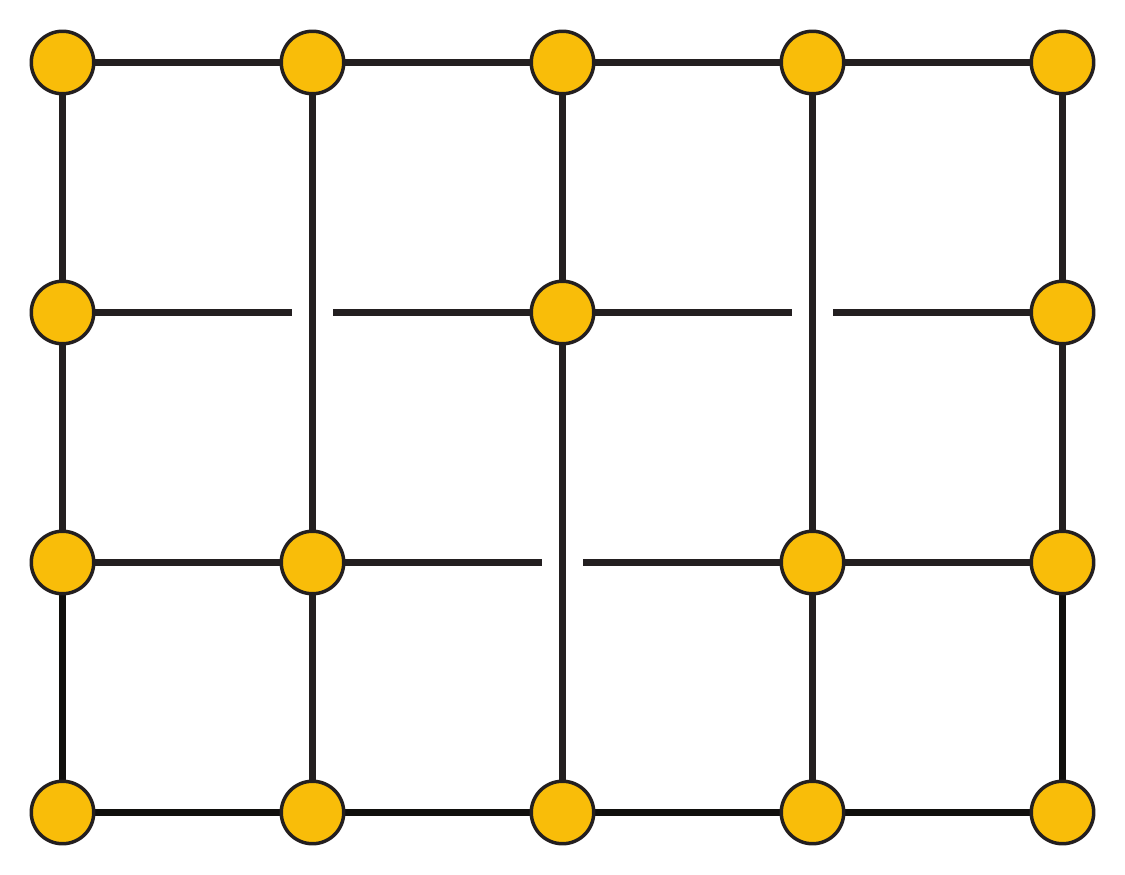}
\caption{A $1$-planar graph: each edge is crossed at most once. Although more general than planar graphs, this class of graphs does not adequately model real-world road networks.}
\label{fig:1planar}
\end{figure}

Among the many generalizations of planar graphs that have been studied in the graph theory literature, another one seems more promising as a model for road networks: the $1$-planar graphs~\cite{Rin-AMSUH-65} or more generally $k$-planar graphs~\cite{PacTot-Combinatorica-97,GriBod-Algorithmica-07,dujmovic2015genus}. A $1$-planar graph is a graph in which every road segment has at most one crossing (Figure~\ref{fig:1planar}). More generally a $k$-planar graph is a graph in which every road segment has at most $k$ crossings.  Many of the sparsity properties of these graphs follow directly from planarization: if one replaces each crossing with a vertex, one obtains a planar graph in which the number of vertices has been blown up only by a factor of $O(k)$. Based on this principle,
it is known that these form sparse families of graphs, and in particular they obey a separator theorem like that for planar graphs but with a dependence on $k$ as well as on the number of vertices in the size of the separator~\cite{dujmovic2015genus}.
Although $k$-planar graphs are NP-hard to recognize from their graph structure alone~\cite{GriBod-Algorithmica-07,BanCabEpp-WADS-13}, that is not problematic for their application to road networks, because in this case an embedding with few crossings would already be known: the actual embedding of the roads on the surface of the earth.

Therefore, it is tempting to model road networks as 1-planar or $k$-planar graphs. However, the restriction on the number of crossings per edge may be too restrictive to model real-world road network graphs. As we show, the assumption that road networks are $1$-planar does not fit the actual data, because  real-world road network data includes road segments that have many crossings. In particular, a long segment of highway may have many crossings between interchanges. However, despite the poor fit of this model to the data, we may take inspiration from $1$-planar graphs in finding a more general class of graphs with few crossings per edge in some more general sense, that still maintains the other desirable properties of this graph class and that allows the algorithms from the theory of $1$-planar graphs to be applied to road networks.

\section{Overview of new results}
\subsection{The crossing graph}

The main new idea of this paper is to study crossings in road networks by introducing a new auxiliary graph, the \emph{crossing graph} of the road network.

\begin{definition}
\normalfont
We define the \emph{crossing graph} of an embedded graph $G$ to be an undirected graph, different from $G$ itself.
Each edge of $G$ becomes a vertex in the crossing graph.
When two edges of $G$ cross each other.
we connect the two corresponding vertices of the crossing graph (the ones representing the two segments) by an edge representing the crossing.
\end{definition}

This concept is illustrated in Figure~\ref{fig:crossing-graph}.
In particular, for a road network, each vertex of the crossing graph represents a segment of road, and each edge of the crossing graph represents two road segments that cross without meeting at an intersection.

\begin{figure}[t]
\centering
\includegraphics[width=0.65\columnwidth]{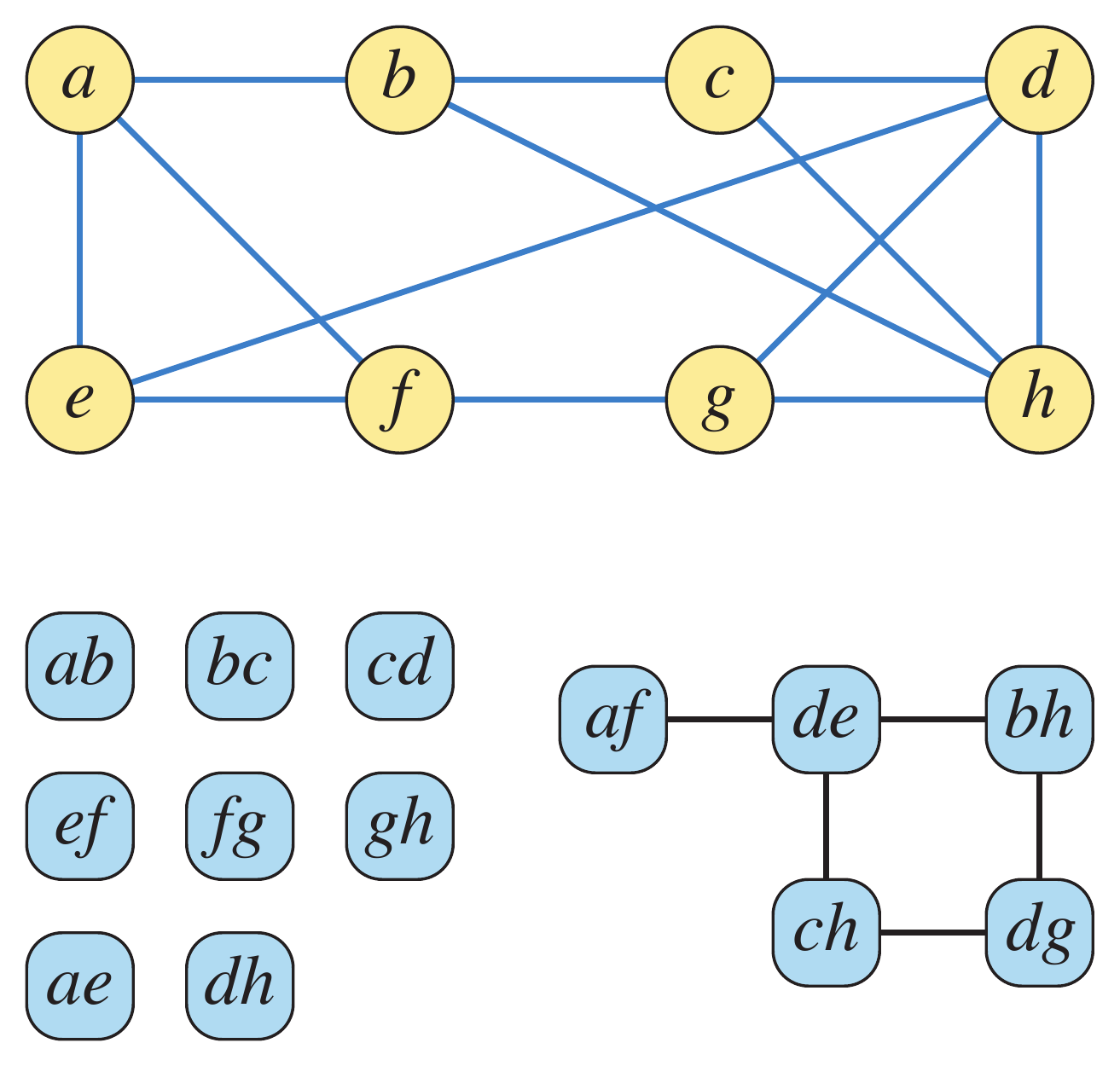}
\caption{A drawing of a graph with crossings (top) and its crossing graph (bottom). The eight isolated vertices of the crossing graph correspond to the eight uncrossed edges in the upper graph.}
\label{fig:crossing-graph}
\end{figure}

This structure allows many natural properties of the underlying road network can be read off directly from the crossing graph. For instance, a segment of road is uncrossed if it corresponds to an isolated vertex (one without any incident edges) in the crossing graph.
A road network is crossing-free (planar) if and only if all its segments are uncrossed, if and only if the vertices of the crossing graph are all isolated, if and only if the crossing graph itself is an independent set. 

For a more complex example, a road network can be modeled as a $1$-planar graph (each edge has at most one crossing) if and only if the crossing graph has maximum degree one; that is, if its crossing graph is a matching. Similarly, the road network is $k$-planar if its crossing graph has maximum degree~$k$.

Our hypothesis is that the crossing graph is a sparse graph (although possibly not one with constant maximum degree) and that its structure can be used to investigate the graph-theoretic structure of the road network itself. We study this question both empirically (by computing and examining the structure of crossing graphs for actual large-scale road networks) and theoretically (by proving that certain types of sparsity in the crossing graph imply the existence of small separators and efficient algorithms for the underlying road networks).

\subsection{Empirical experiments}

We investigate empirically the graph structure of the crossing graph, by constructing this graph for the 80 of the most populated urban areas in the world given by the Urban Road Network Dataset~\cite{urban}.


Our investigations show that  the degree of the crossing graph is small, but not necessarily small enough to consider to be a constant: we found vertices of degree up to 166 (road segments with up to 166 other segments crossing them, none of the crossings forming an intersection). However, we found that the crossing graphs tend to have much smaller \emph{degeneracy}, a number $d$ such that every graph in a given family of graphs (closed under taking subgraphs) has at least one vertex of degree at most~$d$. For instance, trees are exactly the connected graphs of degeneracy one, and our work found that many of the connected components of the crossing graph are trees. More generally, the maximum degeneracy that we found in the crossing graph of any road network was~$6$.

\subsection{Theory of networks with\\ sparse crossing graphs}
Based on our empirical investigations, we undertook a theoretical study of the networks whose crossing graphs have bounded degeneracy. We prove theoretical results showing that these networks are closed under subgraphs and that (like the $k$-planar graphs) they always have small separators. Therefore, they form a new family of graphs of polynomial expansion.

\section{Preliminaries}

\begin{figure*}[t]
\centering
\includegraphics[width=0.9\textwidth]{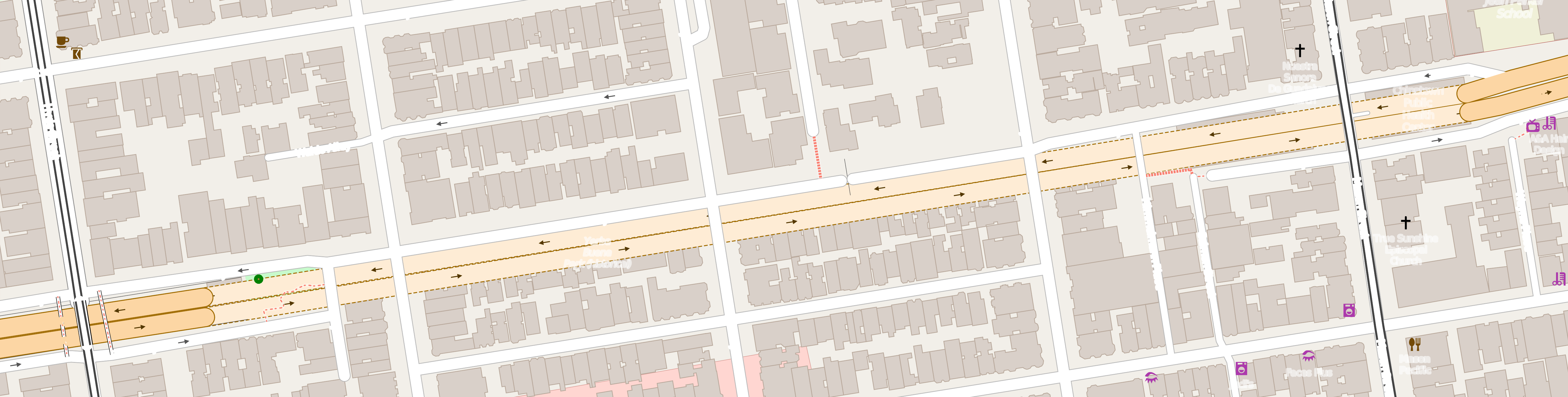}
\caption{The Robert C. Levy tunnel in San Francisco, in which Broadway (orange) passes under seven other streets (L--R: Hyde, Cyrus Pl., Leavenworth, Jones, Taylor, Himmelmann Pl., and Mason) without intersecting them. CC-BY-SA image from OpenStreetMap.}
\label{fig:BroadwayTunnel}
\end{figure*}

Before detailing our experimental and theoretical results, we provide some necessary definitions.

\subsection{Sparse graph properties}
\label{sec:sparsedef}

All graphs in this paper finite. Although road networks are typically directed (by the direction of traffic on one-way streets and divided highways), the direction of the edges does not matter for the crossing pattern and separator properties considered here. Therefore, we treat these graphs as undirected.

\begin{definition}
\normalfont
The \emph{degree} of a vertex in a graph is the number of edges touching that vertex. The \emph{minimum degree} $\delta(G)$ and \emph{maximum degree} $\Delta(G)$ of a graph $G$ are the minimum and maximum, respectively, of the degrees of the vertices in $G$.
A family $\mathscr{F}$ of graphs has \emph{bounded degree} if all graphs in $\mathscr{F}$ have maximum degree $O(1)$; that is, if there is an upper bound on the maximum degree that may depend on $\mathscr{F}$ itself but that does not depend on the choice of a graph within $\mathscr{F}$.
\end{definition}

\begin{definition}
\normalfont
The \emph{degeneracy} of a graph $G$ is the maximum, over subgraphs of $G$, of the minimum degree of the subgraph.
\end{definition}

A concept equivalent to degeneracy (but differing from it by one) was originally called the coloring number by Erd{\H{o}s and Hajnal~\cite{ErdHaj-AMH-66}.
For instance, the graphs of degeneracy one are exactly the forests. A graph has degeneracy at most $d$ if and only if its vertices can be ordered in such a way that every vertex has at most $d$ later neighbors. For, given such an ordering, every subgraph of $G$ has a vertex of degree at most $d$, namely the first vertex of the subgraph to appear in the ordering. Given a graph of degeneracy $d$, an ordering with this property can be found by greedily removing the minimum degree vertex from each remaining subgraph. This greedy removal process can be performed in linear time and allows the degeneracy to be computed in linear time~\cite{MatBec-JACM-83}.

\begin{definition}
As with degree, we define a family $\mathscr{F}$ of graphs to have \emph{bounded degeneracy} if all graphs in $\mathscr{F}$ have degeneracy $O(1)$.
\end{definition}

That is, family $\mathscr{F}$ has bounded degeneracy if there is an upper bound on the degeneracy of the graphs in $\mathscr{F}$ that may depend on $\mathscr{F}$ itself but that does not depend on the choice of a graph within $\mathscr{F}$.

\begin{definition}
\normalfont
The \emph{hop count} from a vertex $u$ to vertex $v$ in a graph is the minimum number of edges between them. The \emph{radius} of a graph or subgraph is the smallest number $r$ such that there exists a vertex within hop count of $r$ of all other vertices.
An $r$-\emph{shallow minor} of a graph $G$ is a graph obtained from $G$ by possibly deleting some edges and/or vertices of $G$, and then contracting some radius-$r$ subgraphs of the remaining graph into supervertices.
A family $\mathscr{F}$ of graphs has \emph{bounded expansion} if, for all choices of the parameter $r$,
the $r$-shallow minors of the graphs in $\mathscr{F}$ have bounded edge/vertex ratio.
More strongly, family $\mathscr{F}$ has \emph{polynomial expansion} if this edge/vertex ratio is bounded by a polynomial in $r$.
\end{definition}

In particular, for $r=0$ the shallow minors are just the subgraphs, so a family of graphs with bounded expansion or polynomial expansion must have subgraphs with bounded edge/vertex ratio. This implies that they necessarily also have bounded degeneracy. The graphs of polynomial expansion include the $k$-apex graphs, $k$-genus graphs, and $k$-planar graphs~\cite{NevOss-Sparsity-12}, described in our earlier discussion of near-planar families of graphs (Section~\ref{sec:past-nearly-planar}).

Graph families of polynomial expansion can also be characterized in terms of separators, small sets of vertices that partition the graph and form the basis of many divide-and-conquer graph algorithms.

\begin{definition}
\normalfont
For an $n$-vertex graph $G$ and a constant $c<1$ we define a $c$-separator to be a subset $S$ of vertices of $G$ such that every connected component of $G\setminus S$ (the subgraph formed by deleting $S$ from $G$) has at most $cn$ vertices. We say that a family $\mathscr{F}$ of graphs has \emph{sublinear separators} if there exist constants $c$, $d$, and $e$, with $c<1$ and $e<1$, such that every $n$-vertex graph in $\mathscr{F}$ has a $c$-separator of size (number of vertices) at most $dn^e$.
\end{definition}

For instance, the famous planar separator theorem states that planar graphs have $2/3$-separators of size $O(\sqrt n)$, so we can take $c=2/3$, $d=O(1)$, and $e=1/2$~\cite{LipTar-SJAM-79}. A \emph{separator hierarchy} is formed by taking separators recursively, until all remaining components have size $O(1)$; for planar graphs it can be constructed in linear time~\cite{Goo-JCSS-95}, and enables linear-time computation of shortest paths~\cite{KleRaoRau-STOC-94}, among other problems.

Then, for a family $\mathscr{F}$ of graphs that is closed under taking subgraphs (every subgraph of a graph in $\mathscr{F}$ is also in $\mathscr{F}$), $\mathscr{F}$ has polynomial expansion if and only if $\mathscr{F}$ has sublinear separators~\cite{DvoNor-SSS-15}. The graphs of polynomial expansion have other important algorithmic properties, not directly deriving from their separators; for instance, for every fixed pattern graph $H$ it is possible to test whether $H$ is a subgraph of a graph in a family of bounded expansion (the subgraph isomorphism problem) in linear time. More generally, one can test any property that can be formulated in the first-order logic of graphs, for members of a family of graphs of bounded expansion, in linear time~\cite{NevOss-Sparsity-12}.

\subsection{Classification of nonplanarities}

We distinguish between two kinds of nonplanarity in a road network.

\begin{definition}
\normalfont
An \emph{embedding} of a graph is a mapping from its vertices to points and its edges to curves,
such that the vertices at the ends of each edge are mapped to the points at the ends of the corresponding curve. A \emph{crossing} is a point where two edge curves intersect that is not a common endpoint of both curves.
A \emph{removable crossing} is a crossing between two edges in an embedding of a graph that
can be removed (without introducing other crossings) by making only local reroutings in the embedding.
\end{definition}

Such a crossing may occur, for instance, when an actual segment of road follows a curved path, passing between segments of other roads without crossing them, but the network segment representing it follows a different straight path that crosses nearby road segments. In such a case, re-routing the segment to follow the actual path of the road will cause these crossings to go away. 

\begin{definition}
\normalfont
An \emph{essential crossing} is a crossing between two edges in an embedding of a graph
that represent disjoint (non-intersecting) road segments and that cannot be removed by local changes.
\end{definition}

Such a crossing may be caused, for instance, when one road follows a tunnel that crosses under several other roads (Figure~\ref{fig:BroadwayTunnel}), or by the multiple crossing segments of a highway interchange (Figure~\ref{fig:HighFive}).

\begin{figure}[t]
\centering
\includegraphics[width=0.75\columnwidth]{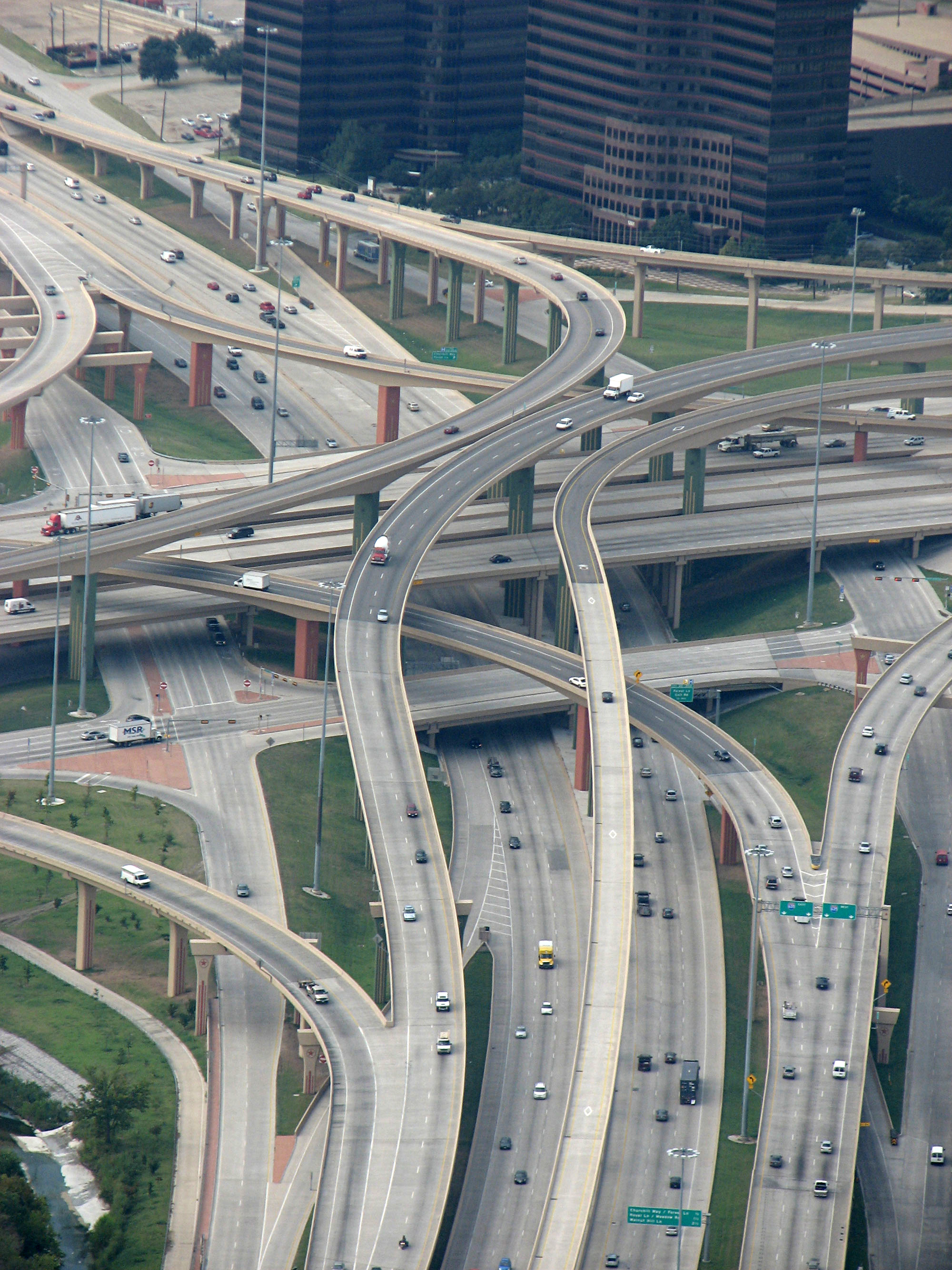}
\caption{High Five Interchange in Dallas, Texas. CC-BY image File:High Five.jpg by fatguyinalittlecoat from Wikimedia commons.}
\label{fig:HighFive}
\end{figure}

It is possible to remove the essential crossings of a road network by replacing each crossing by an artificial intersection point, one that exists in the graph but not in the actual road network that it represents. We call a planar network constructed in this way the \emph{planarization} of the road network. Indeed, the TIGER data set, commonly used for experiments on road networks, has been planarized in this way. If the edges in this data set are drawn as straight line segments,
instead of following the curves of the actual roads that the edges represent, some crossings may arise from the straightening, but the underlying graph of this data set is planar. For this reason,
our experiments use a different data set, the Urban Road Network Dataset, that has not already been planarized~\cite{urban}.

However, planarized networks cannot be used to obtain correct results for many road network computations. For instance, using the planarization of a road network in a shortest path routing algorithm would potentially create routes that turn from one road to another at the artificial intersections added in planarization, which do not form usable routes in the real-world network.
For this reason, it is problematic to use algorithms designed specifically to work in planar networks, such as the known linear-time planar graph shortest path algorithms~\cite{KleRaoRau-STOC-94}, on road network data. Instead, we must show that despite their non-planarities, road networks have the underlying properties necessary to support generalized versions of these efficient algorithms.


\section{Experiments}
In this section, we examine the sparsity of crossing graphs experimentally, on real-world road networks. We compute and analyze the crossing graphs for the $80$ of the most populated urban areas in the world. Our analysis uses the Urban Road Network Dataset~\cite{urban}.


\begin{table*}
\centering
\caption{Crossing Graphs (both essential and removable crossings)}
 \begin{tabular}{|c|c|c|c|c|c|c|c|c|} \hline
City & Roads & Crossings & Uncrossed roads & Degeneracy & Degree & Components & Trees & Non-trees \\ \hline
Abidjan & 46423 & 699 & 45565 & 3 & 11 & 214 & 193 & 21 \\ \hline
Ahmedabad & 18327 & 703 & 17567 & 3 & 64 & 149 & 122 & 27 \\ \hline
Ankara & 111054 & 2713 & 107943 & 3 & 16 & 769 & 643 & 126 \\ \hline
Atlanta & 381649 & 7505 & 372135 & 4 & 33 & 2737 & 2501 & 236 \\ \hline
Bangdung & 26154 & 537 & 25555 & 2 & 23 & 156 & 140 & 16 \\ \hline
Bangkok & 188069 & 9940 & 178225 & 5 & 38 & 1991 & 1637 & 354 \\ \hline
Barcelona & 296969 & 13218 & 281762 & 4 & 38 & 3792 & 3309 & 483 \\ \hline
Beijing & 110115 & 13491 & 98339 & 5 & 71 & 2440 & 1915 & 525 \\ \hline
BeloHorizonte & 94207 & 2797 & 91221 & 4 & 23 & 632 & 531 & 101 \\ \hline
Bengaluru & 200828 & 3102 & 197273 & 4 & 36 & 946 & 846 & 100 \\ \hline
Bogota & 199846 & 5929 & 194621 & 5 & 166 & 1070 & 960 & 110 \\ \hline
Boston & 470360 & 9324 & 458955 & 5 & 25 & 3125 & 2830 & 295 \\ \hline
BuenosAires & 435039 & 5344 & 429340 & 6 & 60 & 1284 & 1083 & 201 \\ \hline
Calcutta & 90283 & 991 & 88995 & 3 & 18 & 395 & 374 & 21 \\ \hline
Chengdu & 25847 & 3278 & 23076 & 4 & 100 & 530 & 412 & 118 \\ \hline
Chongqing & 27251 & 3899 & 23207 & 4 & 18 & 873 & 658 & 215 \\ \hline
Dalian & 18017 & 1169 & 16638 & 3 & 21 & 352 & 292 & 60 \\ \hline
Dallas & 548721 & 15287 & 532025 & 4 & 52 & 3882 & 3136 & 746 \\ \hline
Delhi & 76240 & 1853 & 74240 & 3 & 24 & 458 & 364 & 94 \\ \hline
Detroit & 390597 & 8163 & 381056 & 4 & 34 & 2374 & 2070 & 304 \\ \hline
Dhaka & 19922 & 477 & 19445 & 2 & 20 & 128 & 120 & 8 \\ \hline
Dongguan & 11076 & 1773 & 9557 & 4 & 28 & 279 & 183 & 96 \\ \hline
Fuzhou & 16093 & 1948 & 14496 & 4 & 23 & 330 & 247 & 83 \\ \hline
Guangzhou & 69410 & 9765 & 60238 & 4 & 26 & 1804 & 1234 & 570 \\ \hline
Hangzhou & 28270 & 3433 & 25238 & 4 & 26 & 623 & 459 & 164 \\ \hline
Harbin & 14805 & 1383 & 13465 & 3 & 38 & 271 & 195 & 76 \\ \hline
HoChiMinh & 110953 & 1627 & 109087 & 3 & 13 & 519 & 469 & 50 \\ \hline
Houston & 632289 & 13337 & 616864 & 4 & 28 & 3888 & 3255 & 633 \\ \hline
Hyderabad & 172822 & 2461 & 169948 & 3 & 166 & 737 & 701 & 36 \\ \hline
Istambul & 380466 & 12002 & 368090 & 4 & 56 & 2480 & 2059 & 421 \\ \hline
Jacarta & 105318 & 4811 & 100998 & 4 & 159 & 746 & 620 & 126 \\ \hline
Johannesburg & 267193 & 6715 & 259424 & 4 & 38 & 1957 & 1710 & 247 \\ \hline
Karachi & 55362 & 1091 & 54236 & 3 & 38 & 235 & 206 & 29 \\ \hline
Lahore & 41115 & 1506 & 39689 & 3 & 61 & 223 & 188 & 35 \\ \hline
London & 376437 & 9015 & 365645 & 4 & 23 & 2988 & 2757 & 231 \\ \hline
LosAngeles & 552690 & 17905 & 532472 & 5 & 50 & 4731 & 4140 & 591 \\ \hline
Madrid & 511910 & 23482 & 487090 & 5 & 49 & 5139 & 4349 & 790 \\ \hline
Manila & 328623 & 5518 & 321984 & 4 & 26 & 1786 & 1664 & 122 \\ \hline
Medellin & 33310 & 1118 & 32174 & 4 & 20 & 271 & 225 & 46 \\ \hline
Miami & 312082 & 6396 & 304945 & 5 & 41 & 1668 & 1387 & 281 \\ \hline
Milan & 262466 & 8512 & 252789 & 4 & 25 & 2536 & 2227 & 309 \\ \hline
Moscow & 940251 & 20469 & 916398 & 4 & 40 & 6647 & 6185 & 462 \\ \hline
Mumbai & 61299 & 2085 & 58976 & 3 & 20 & 580 & 514 & 66 \\ \hline
Nanjing & 29947 & 3224 & 27018 & 5 & 21 & 623 & 467 & 156 \\ \hline
Naples & 169589 & 7089 & 161590 & 4 & 49 & 2242 & 2006 & 236 \\ \hline
\end{tabular}
\label{table:newcrossing} 
\end{table*}


\begin{table*}
\centering
\caption{Crossing Graphs (essential crossings only)}
 \begin{tabular}{|c|c|c|c|c|c|c|c|c|} \hline
City & Roads & Crossings & Uncrossed roads & Degeneracy & Degree & Components & Trees & Non-trees \\ \hline
Abidjan & 46423 & 175 & 46218 & 2 & 5 & 52 & 41 & 11 \\ \hline
Ahmedabad & 18327 & 299 & 18056 & 3 & 10 & 46 & 24 & 22 \\ \hline
Ankara & 111054 & 1041 & 110053 & 3 & 9 & 233 & 144 & 89 \\ \hline
Atlanta & 381649 & 2081 & 379231 & 4 & 14 & 664 & 533 & 131 \\ \hline
Bangdung & 26154 & 175 & 25976 & 2 & 21 & 45 & 37 & 8 \\ \hline
Bangkok & 188069 & 6902 & 181482 & 5 & 35 & 1252 & 956 & 296 \\ \hline
Barcelona & 296969 & 5048 & 291548 & 4 & 36 & 1380 & 1118 & 262 \\ \hline
Beijing & 110115 & 10772 & 101255 & 5 & 71 & 1764 & 1314 & 450 \\ \hline
BeloHorizonte & 94207 & 1047 & 93155 & 3 & 11 & 241 & 187 & 54 \\ \hline
Bengaluru & 200828 & 1158 & 199700 & 3 & 35 & 249 & 190 & 59 \\ \hline
Bogota & 199846 & 3270 & 197517 & 5 & 18 & 422 & 357 & 65 \\ \hline
Boston & 470360 & 3668 & 466144 & 5 & 23 & 1098 & 916 & 182 \\ \hline
BuenosAires & 435039 & 3217 & 431874 & 6 & 60 & 715 & 559 & 156 \\ \hline
Calcutta & 90283 & 363 & 89894 & 2 & 15 & 85 & 70 & 15 \\ \hline
Chengdu & 25847 & 2385 & 23951 & 4 & 99 & 346 & 257 & 89 \\ \hline
Chongqing & 27251 & 2163 & 25020 & 4 & 11 & 522 & 368 & 154 \\ \hline
Dalian & 18017 & 757 & 17126 & 3 & 21 & 230 & 186 & 44 \\ \hline
Dallas & 548721 & 7694 & 541247 & 4 & 40 & 1553 & 966 & 587 \\ \hline
Delhi & 76240 & 836 & 75437 & 3 & 10 & 171 & 100 & 71 \\ \hline
Detroit & 390597 & 3123 & 387062 & 4 & 10 & 891 & 709 & 182 \\ \hline
Dhaka & 19922 & 247 & 19711 & 2 & 7 & 61 & 56 & 5 \\ \hline
Dongguan & 11076 & 1162 & 10133 & 4 & 26 & 171 & 101 & 70 \\ \hline
Fuzhou & 16093 & 1228 & 15134 & 3 & 14 & 194 & 130 & 64 \\ \hline
Guangzhou & 69410 & 6817 & 63105 & 4 & 25 & 1241 & 752 & 489 \\ \hline
Hangzhou & 28270 & 2462 & 26166 & 3 & 23 & 438 & 305 & 133 \\ \hline
Harbin & 14805 & 896 & 13973 & 3 & 38 & 164 & 113 & 51 \\ \hline
HoChiMinh & 110953 & 616 & 110325 & 3 & 12 & 149 & 116 & 33 \\ \hline
Houston & 632289 & 6769 & 625518 & 4 & 19 & 1431 & 894 & 537 \\ \hline
Hyderabad & 172822 & 998 & 171967 & 2 & 166 & 103 & 89 & 14 \\ \hline
Istambul & 380466 & 4391 & 376409 & 4 & 31 & 845 & 573 & 272 \\ \hline
Jacarta & 105318 & 2338 & 103588 & 4 & 158 & 247 & 171 & 76 \\ \hline
Johannesburg & 267193 & 1871 & 265257 & 3 & 38 & 509 & 379 & 130 \\ \hline
Karachi & 55362 & 460 & 54931 & 3 & 28 & 97 & 76 & 21 \\ \hline
Lahore & 41115 & 523 & 40632 & 3 & 57 & 71 & 49 & 22 \\ \hline
London & 376437 & 4946 & 370890 & 4 & 23 & 1497 & 1350 & 147 \\ \hline
LosAngeles & 552690 & 8941 & 543170 & 5 & 16 & 2095 & 1683 & 412 \\ \hline
Madrid & 511910 & 8734 & 503015 & 5 & 47 & 1981 & 1598 & 383 \\ \hline
Manila & 328623 & 2276 & 326260 & 4 & 19 & 577 & 505 & 72 \\ \hline
Medellin & 33310 & 635 & 32765 & 3 & 12 & 114 & 82 & 32 \\ \hline
Miami & 312082 & 2707 & 309459 & 5 & 13 & 529 & 294 & 235 \\ \hline
Milan & 262466 & 4297 & 258068 & 4 & 25 & 1103 & 887 & 216 \\ \hline
Moscow & 940251 & 9686 & 930967 & 4 & 37 & 2249 & 1947 & 302 \\ \hline
Mumbai & 61299 & 1218 & 60081 & 3 & 20 & 274 & 224 & 50 \\ \hline
Nanjing & 29947 & 2184 & 28039 & 5 & 19 & 391 & 266 & 125 \\ \hline
Naples & 169589 & 3682 & 165822 & 4 & 48 & 1024 & 889 & 135 \\ \hline
\end{tabular}
\label{table:npcrossing} 
\end{table*}

The Urban Road Network Dataset includes graphs with self-loops and parallel edges. We removed self-loops and parallel edges before processing the data. The data contains both essential and removable crossings.

To find the crossings in the Urban Road Network Dataset, we used a plane sweep algorithm for line segment intersection detection~\cite{BoiPre-SICOMP-00}, as implemented in the CGAL computational geometry library~\cite{CGAL}. Although theoretically-faster algorithms are known for the type of data considered here, in which the segments form a connected geometric graph with few crossings~\cite{eppstein2009linear}, CGAL's plane sweep is practical and usable, and the slight superlinearity of its time bound is not problematic for the problem sizes we tested.

Because a perfect identification of essential crossings would require the solution of the NP-hard problem of minimizing the number of crossings in graph embeddings~\cite{GarJoh-SJADM-83}, our experiments determined whether a crossing is essential or removable heuristically, by using the fact that the data associated with each road segment in the Urban Road Network Dataset indicates whether it is a bridge or tunnel. Our heuristic is that when a crossing occurs between two road segments neither of which is a bridge or a tunnel, then it is removable. However, the  Urban Road Network Dataset only includes this bridge and tunnel labeling for a subset of the cities that it covers. For this reason, we restricted our experiments to this subset.

We used the NetworkX Python package~\cite{NetworkX} to study the structure of the crossing graphs  we constructed.

\subsection{Hypothesis}

Based on our intuitions concerning bridges and tunnels in real-world road networks,
we expected the crossing graphs to include some vertices of moderate degree, but otherwise to be very sparse. For instance, we considered it to be possible that all of the connected components of the crossing graph would be trees.

\subsection{Results}

The results of our experiments can be seen in Table \ref{table:newcrossing} and \ref{table:npcrossing}.
We have given some of the key properties of the crossing graphs of road networks for $45$ out of $80$ cities in the table. The table columns labeled ``Roads'' and ``Crossings'' give the total number of nodes and edges, respectively, in each crossing graph; that is, the numbers of road segments and crossings in the original road network. The column giving the number of uncrossed roads lists the number of nodes in the crossing graph which have no incident edges; that is, the number of uncrossed road segments. These isolated nodes constitute the vast majority of crossing graph nodes. Table~\ref{table:newcrossing} captures both essential and removable crossings whereas Table~\ref{table:npcrossing} captures only essential crossings.

The columns labeled ``Degeneracy'' and ``Degree'' give the degeneracy and maximum degree, respectively, of the crossing graph, as defined in Section~\ref{sec:sparsedef}. The remaining columns detail the total number of nontrivial connected components, the number of components that are trees, and the number of  components that are not trees, respectively, in the crossing graph.

\subsection{Analysis}

The table show that indeed these graphs are sparse. More specifically, our hypothesis that the degeneracy would be significantly smaller than the maximum degree held up in the experiments.

Although it is not true that (as we hypothesized) all components of the crossing graphs are trees, most of them are. The remaining non-tree components all have low degeneracy (at most $6$).


\section{Theoretical Analysis of Graphs with Sparse Crossings}

Both our experimental results, and our consideration of nonplanar tunnels and interchanges forming essential crossings in (non-planarized) real road networks, motivate the following question. Suppose that, as the experiments appear to show, road networks have sparse crossing graphs. More precisely, suppose that their crossing graphs have bounded degeneracy, but not necessarily bounded degree. What does this imply about the graph-theoretic properties of road networks? Do they have bounded degeneracy? Polynomial expansion? Here we give positive answers to both of these questions.

First, let us formalize the notion of a graph with a sparse crossing graph.

\begin{definition}
\normalfont
We define a \emph{nice embedding} to be a mapping of the vertices of the graph to points in the plane, and the edges to curves, such that the following conditions are all met:
\begin{itemize}
\item Each edge is mapped to a Jordan arc (a non-self-intersecting curve) whose endpoints are the images of the endpoints of the edge.
\item If an edge and a vertex are disjoint in the graph, their images in the plane are disjoint.
\item If two edges are mapped to curves that intersect, then that intersection consists of a single point, and is either a shared endpoint of both edges or a point where their two curves cross.
\item No three edges have curves that all cross at the same point.
\end{itemize}
\end{definition}

Let $\mathscr{C}_d$ denote the family of embedded graphs with nice embeddings, such that the crossing graphs of these embeddings have degeneracy at most~$d$. 

\begin{definition}
We say that a graph $G$ is \emph{$d$-crossing-degenerate} if it belongs to $\mathscr{C}_d$.
\end{definition}

In order to apply the equivalence between separator theorems and polynomial expansion of Dvo{\v{r}}{\'a}k and Norin~\cite{DvoNor-SSS-15}, we need to verify that the class of graphs we care about is closed under subgraphs.

\begin{lemma}
\label{lem:subgraph-closure}
Every subgraph of a graph in $\mathscr{C}_d$ also belongs to $\mathscr{C}_d$.
\end{lemma}

\begin{proof}
Let $G$ be in $\mathscr{C}_d$ and let $H$ be a subgraph of $G$. Embed $H$ by deleting the edges and vertices of $G\setminus H$ from the embedding of $G$. Then the crossing graph of the resulting embedding of $H$ is an induced subgraph of the crossing graph of the embedding of $G$.
Since taking induced subgraphs cannot increase the degeneracy, the degeneracy of the crossing graph of $H$ is at most~$d$. Therefore, $H$ belongs to $\mathscr{C}_d$.
\end{proof}

Next, we examine the number of crossings that a graph in this family can have. As we show below, a linear bound on the crossing number follows directly from the assumption of low crossing graph degeneracy.

\begin{lemma}
\label{lem:few-crossings}
For the graphs in $\mathscr{C}_d$, a graph with $m$ edges has at most $dm$ crossings.
\end{lemma}

\begin{proof}
We can reduce any graph in $\mathscr{C}_d$ to one with no edges (and no crossings) by repeatedly removing an edge that is crossed by at most $d$ other edges, using the assumption that the crossing graph is $d$-degenerate and therefore that there exists an edge that is crossed at most $d$ times. This process eliminates at most $d$ crossings per step and takes $m$ steps, so there are at most $dm$ crossings in the given graph.
\end{proof}

In contrast, graphs with many edges are known to have many crossings per edge. We use the following well-known \emph{crossing number inequality} of Ajtai, Chv\'atal, Newborn, Szemer\'edi, and Leighton~\cite{AjtChvNew-TPC-82,Lei-VLSI-83,PacRadTar-DCG-06}:

\begin{lemma}
\label{lem:crossing-inequality}
Let $G$ be an embedded graph with $n$ vertices and $m$ edges, with $m\ge 4n$. Then $G$ has $\Omega(m^3/n^2)$ crossings.
\end{lemma}

This allows us to show that the $d$-crossing-degenerate graphs are sparse.

\begin{lemma}
\label{lem:sparse}
Every $n$-vertex embedded graph $G$ in $\mathscr{C}_d$ has $O(n\sqrt d)$ edges.
\end{lemma}

\begin{proof}
Let the number of edges in $G$ be $\gamma n\sqrt d$, for some parameter $\gamma$.
Then by Lemma~\ref{lem:crossing-inequality} the number of crossings is $\Omega(n(\gamma\sqrt d)^3)$ and the number of crossings per edge is $\Omega((\gamma\sqrt d)^2)=\Omega(\gamma^2 d)$. However, by Lemma~\ref{lem:few-crossings} the number of crossings per edge is also at most~$d$. For these two things to both be true, it must be the case that $\gamma=O(1)$, so the number of edges in $G$ is $O(n\sqrt d)$.
\end{proof}

\begin{definition}
\normalfont
Given a nicely embedded graph $G$, we define the \emph{planarization} $P(G)$ to be the planar graph that has a vertex for each vertex or crossing point of $G$, and an edge for each maximal segment of an edge curve of $G$ that does not contain a crossing point.
\end{definition}

It follows from our previous lemmas that $d$-crossing-degen\-er\-ate graphs have small planarizations.

\begin{lemma}
\label{lem:small-planarization}
For the graphs in $\mathscr{C}_d$, every $n$-vertex graph $G$ has a planarization with $O(nd^{3/2})$ vertices and edges.
\end{lemma}

\begin{proof}
This follows immediately from the already-proven facts that $G$ itself has $O(n\sqrt d)$ edges (Lemma~\ref{lem:sparse}) and an average of at most $d$ crossings per edge (Lemma~\ref{lem:few-crossings}).
\end{proof}

Using these lemmas, we can prove our main result, that these graphs have sublinear separators.

\begin{theorem}
The graphs in $\mathscr{C}_d$ have sublinear separators and polynomial expansion.
A sublinear separator hierarchy for these graphs can be constructed from their planarizations in linear time.
\end{theorem}

\begin{proof}
To prove the existence of sublinear separators, we apply the planar separator theorem~\cite{LipTar-SJAM-79} to the planarization $P(G)$ of a graph $G$ in $\mathscr{C}_d$. By Lemma~\ref{lem:small-planarization} $P(G)$ is larger than $G$ by at most a factor depending only on~$d$, so (treating $d$ as a constant for the purposes of $O$-notation) $P(G)$ has separators of size $O(\sqrt n)$. One application of the separator theorem may produce components of $P(G)$ that are larger than the number $n$ of vertices in $G$, but recursively applying the separator theorem a bounded number of times will produce components that are smaller than $n$ by a constant factor. The resulting separators have vertices in $P(G)$, corresponding to crossings in $G$; they can be transformed into separators in $G$ itself by replacing each crossing vertex in the separator by the set of four endpoints of the corresponding crossing edges in $G$. This replacement only increases the separator size by a constant factor.

Since this method uses only planarization and planar separators, we may apply the linear time method for constructing planar separator hierarchies to $P(G)$~\cite{Goo-JCSS-95}, to get a separator hierarchy to $G$ as well.

The fact that these graphs have polynomial expansion follows from the existence of sublinear separators, and from the fact that $\mathscr{C}_d$ is closed under subgraphs (Lemma~\ref{lem:subgraph-closure}).
\end{proof}

By applying the method of \cite{KleRaoRau-STOC-94}, we obtain:

\begin{corollary}
If we are given the planarization of a graph $G$ in $\mathscr{C}_d$, we can compute shortest paths in $G$ itself in linear time.
\end{corollary}

In conjunction with known fast algorithms for finding planarizations~\cite{eppstein2009linear}, this leads to a linear-time algorithm for shortest paths whenever the number of crossings is sufficiently smaller than the overall number of road segments (as it was in our experiments).

\section{Conclusions}

We have performed a computational study of the removable crossings in large-scale planarized road network data. Our study shows that these crossings form a crossing graph that has high degree vertices (up to degree 166), but that most connected components of the crossing graph are trees and that the few remaining components have maximum degeneracy six. 


Based on our study, we developed a model of nearly-planar graph, the $d$-crossing-degenerate graphs, consisting of the graphs that can be embedded with $d$-degenerate crossing graphs. We showed that this family of graphs is closed under the operation of taking subgraphs. In addition, for constant values of~$d$, these graphs have a linear number of crossings, a linear number of edges, and separators of size proportional to the square root of the number of vertices. In addition, a separator hierarchy for these graphs can be constructed in linear time, and applied in separator-based divide and conquer algorithms for shortest paths and other computational problems on road networks.



\subsection*{Acknowledgments}
This work was supported by NSF grants CCF-1618301 and CCF-1616248.

\bibliographystyle{ACM-Reference-Format}
\bibliography{refs}


\begin{thebibliography}{00}


\ifx \showCODEN    \undefined \def \showCODEN     #1{\unskip}     \fi
\ifx \showDOI      \undefined \def \showDOI       #1{#1}\fi
\ifx \showISBNx    \undefined \def \showISBNx     #1{\unskip}     \fi
\ifx \showISBNxiii \undefined \def \showISBNxiii  #1{\unskip}     \fi
\ifx \showISSN     \undefined \def \showISSN      #1{\unskip}     \fi
\ifx \showLCCN     \undefined \def \showLCCN      #1{\unskip}     \fi
\ifx \shownote     \undefined \def \shownote      #1{#1}          \fi
\ifx \showarticletitle \undefined \def \showarticletitle #1{#1}   \fi
\ifx \showURL      \undefined \def \showURL       {\relax}        \fi
\providecommand\bibfield[2]{#2}
\providecommand\bibinfo[2]{#2}
\providecommand\natexlab[1]{#1}
\providecommand\showeprint[2][]{arXiv:#2}

\bibitem[\protect\citeauthoryear{Ajtai, Chv{\'a}tal, Newborn, and
  Szemer{\'e}di}{Ajtai et~al\mbox{.}}{1982}]%
        {AjtChvNew-TPC-82}
\bibfield{author}{\bibinfo{person}{Mikl{\'o}s Ajtai},
  \bibinfo{person}{V{\'a}clav Chv{\'a}tal}, \bibinfo{person}{M. Newborn}, {and}
  \bibinfo{person}{Endre Szemer{\'e}di}.} \bibinfo{year}{1982}\natexlab{}.
\newblock \showarticletitle{{Crossing-free subgraphs}}.
\newblock In \bibinfo{booktitle}{{\em Theory and Practice of Combinatorics}}.
  \bibinfo{series}{North-Holland Mathematics Studies},
  Vol.~\bibinfo{volume}{60}. \bibinfo{publisher}{North-Holland},
  \bibinfo{pages}{9{--}12}.
\newblock


\bibitem[\protect\citeauthoryear{Bannister, Cabello, and Eppstein}{Bannister
  et~al\mbox{.}}{2013}]%
        {BanCabEpp-WADS-13}
\bibfield{author}{\bibinfo{person}{Michael~J. Bannister},
  \bibinfo{person}{Sergio Cabello}, {and} \bibinfo{person}{David Eppstein}.}
  \bibinfo{year}{2013}\natexlab{}.
\newblock \showarticletitle{{Parameterized complexity of 1-planarity}}. In
  \bibinfo{booktitle}{{\em 13th Int. Symp. Algorithms and Data Structures}}
  {\em (\bibinfo{series}{Lect. Notes in Comput. Sci.})},
  Vol.~\bibinfo{volume}{8037}. \bibinfo{publisher}{Springer},
  \bibinfo{pages}{97{--}108}.
\newblock
\showDOI{%
\url{https://doi.org/10.1007/978-3-642-40104-6_9}}


\bibitem[\protect\citeauthoryear{Boissonnat and Preparata}{Boissonnat and
  Preparata}{2000}]%
        {BoiPre-SICOMP-00}
\bibfield{author}{\bibinfo{person}{Jean-Daniel Boissonnat} {and}
  \bibinfo{person}{Franco~P. Preparata}.} \bibinfo{year}{2000}\natexlab{}.
\newblock \showarticletitle{{Robust plane sweep for intersecting segments}}.
\newblock \bibinfo{journal}{{\em SIAM J. Comput.\/}} \bibinfo{volume}{29},
  \bibinfo{number}{5} (\bibinfo{year}{2000}), \bibinfo{pages}{1401{--}1421}.
\newblock
\showDOI{%
\url{https://doi.org/10.1137/S0097539797329373}}


\bibitem[\protect\citeauthoryear{Chung}{Chung}{1990}]%
        {Chu-PFV-90}
\bibfield{author}{\bibinfo{person}{Fan R.~K. Chung}.}
  \bibinfo{year}{1990}\natexlab{}.
\newblock \showarticletitle{{Separator theorems and their applications}}.
\newblock In \bibinfo{booktitle}{{\em Paths, Flows, and VLSI-Layout}},
  \bibfield{editor}{\bibinfo{person}{Bernhard Korte},
  \bibinfo{person}{L{\'a}szl{\'o} Lov{\'a}sz}, \bibinfo{person}{Hans~J{\"u}rgen
  Pr{\"o}mel}, {and} \bibinfo{person}{Alexander Schrijver}} (Eds.).
  \bibinfo{series}{Algorithms and Combinatorics}, Vol.~\bibinfo{volume}{9}.
  \bibinfo{publisher}{Springer-Verlag}, \bibinfo{pages}{17{--}34}.
\newblock
\showISBNx{978-0-387-52685-0}


\bibitem[\protect\citeauthoryear{Demaine, Fomin, Hajiaghayi, and
  Thilikos}{Demaine et~al\mbox{.}}{2005}]%
        {DemFomHaj-JACM-05}
\bibfield{author}{\bibinfo{person}{Erik~D. Demaine}, \bibinfo{person}{Fedor~V.
  Fomin}, \bibinfo{person}{Mohammadtaghi Hajiaghayi}, {and}
  \bibinfo{person}{Dimitrios~M. Thilikos}.} \bibinfo{year}{2005}\natexlab{}.
\newblock \showarticletitle{{Subexponential parameterized algorithms on
  bounded-genus graphs and $H$-minor-free graphs}}.
\newblock \bibinfo{journal}{{\em J. ACM\/}} \bibinfo{volume}{52},
  \bibinfo{number}{6} (\bibinfo{date}{November} \bibinfo{year}{2005}),
  \bibinfo{pages}{866{--}893}.
\newblock
\showDOI{%
\url{https://doi.org/10.1145/1101821.1101823}}


\bibitem[\protect\citeauthoryear{Dujmovi{\'c}, Eppstein, and Wood}{Dujmovi{\'c}
  et~al\mbox{.}}{2015}]%
        {dujmovic2015genus}
\bibfield{author}{\bibinfo{person}{Vida Dujmovi{\'c}}, \bibinfo{person}{David
  Eppstein}, {and} \bibinfo{person}{David~R. Wood}.}
  \bibinfo{year}{2015}\natexlab{}.
\newblock \showarticletitle{{Genus, treewidth, and local crossing number}}.
\newblock In \bibinfo{booktitle}{{\em Proc. 23rd Int. Symp. Graph Drawing and
  Network Visualization}}. \bibinfo{series}{Lect. Notes in Comput. Sci.},
  Vol.~\bibinfo{volume}{9411}. \bibinfo{publisher}{Springer},
  \bibinfo{pages}{87{--}98}.
\newblock


\bibitem[\protect\citeauthoryear{Dvo{\v{r}}{\'a}k and Norin}{Dvo{\v{r}}{\'a}k
  and Norin}{2015}]%
        {DvoNor-SSS-15}
\bibfield{author}{\bibinfo{person}{Zden{\v{e}}k Dvo{\v{r}}{\'a}k} {and}
  \bibinfo{person}{Sergey Norin}.} \bibinfo{year}{2015}\natexlab{}.
\newblock \bibinfo{title}{{Strongly sublinear separators and polynomial
  expansion}}.
\newblock \bibinfo{howpublished}{Electronic preprint arxiv:1504.04821}.
  (\bibinfo{year}{2015}).
\newblock


\bibitem[\protect\citeauthoryear{Eppstein and Goodrich}{Eppstein and
  Goodrich}{2008}]%
        {eppstein2008studying}
\bibfield{author}{\bibinfo{person}{David Eppstein} {and}
  \bibinfo{person}{Michael~T. Goodrich}.} \bibinfo{year}{2008}\natexlab{}.
\newblock \showarticletitle{{Studying (non-planar) road networks through an
  algorithmic lens}}.
\newblock \bibinfo{journal}{{\em Proc. 16th ACM SIGSPATIAL Int. Conf. Advances
  in Geographic Information Systems (ACM GIS 2008)\/}} (\bibinfo{year}{2008}).
\newblock
\showDOI{%
\url{https://doi.org/10.1145/1463434.1463455}}
\newblock
\shownote{Article 16.}


\bibitem[\protect\citeauthoryear{Eppstein, Goodrich, and Strash}{Eppstein
  et~al\mbox{.}}{2009a}]%
        {eppstein2009linear}
\bibfield{author}{\bibinfo{person}{David Eppstein}, \bibinfo{person}{Michael~T.
  Goodrich}, {and} \bibinfo{person}{Darren Strash}.}
  \bibinfo{year}{2009}\natexlab{a}.
\newblock \showarticletitle{{Linear-time algorithms for geometric graphs with
  sublinearly many crossings}}. In \bibinfo{booktitle}{{\em Proc. 20th ACM-SIAM
  Symposium on Discrete Algorithms}}. \bibinfo{publisher}{Society for
  Industrial and Applied Mathematics}, \bibinfo{pages}{150{--}159}.
\newblock


\bibitem[\protect\citeauthoryear{Eppstein, Goodrich, and Trott}{Eppstein
  et~al\mbox{.}}{2009b}]%
        {eppstein2009going}
\bibfield{author}{\bibinfo{person}{David Eppstein}, \bibinfo{person}{Michael~T.
  Goodrich}, {and} \bibinfo{person}{Lowell Trott}.}
  \bibinfo{year}{2009}\natexlab{b}.
\newblock \showarticletitle{{Going off-road: transversal complexity in road
  networks}}. In \bibinfo{booktitle}{{\em Proc. 17th ACM SIGSPATIAL Int. Conf.
  Advances in Geographic Information Systems, Seattle, 2009}}.
  \bibinfo{pages}{23{--}32}.
\newblock
\showDOI{%
\url{https://doi.org/10.1145/1653771.1653778}}


\bibitem[\protect\citeauthoryear{Erd{\H{o}}s and Hajnal}{Erd{\H{o}}s and
  Hajnal}{1966}]%
        {ErdHaj-AMH-66}
\bibfield{author}{\bibinfo{person}{Paul Erd{\H{o}}s} {and}
  \bibinfo{person}{Andr{\'a}s Hajnal}.} \bibinfo{year}{1966}\natexlab{}.
\newblock \showarticletitle{{On chromatic number of graphs and set-systems}}.
\newblock \bibinfo{journal}{{\em Acta Math. Hung.\/}} \bibinfo{volume}{17},
  \bibinfo{number}{1{--}2} (\bibinfo{year}{1966}), \bibinfo{pages}{61{--}99}.
\newblock
\showDOI{%
\url{https://doi.org/10.1007/BF02020444}}


\bibitem[\protect\citeauthoryear{Fabri and Pion}{Fabri and Pion}{2009}]%
        {CGAL}
\bibfield{author}{\bibinfo{person}{Andreas Fabri} {and}
  \bibinfo{person}{Sylvain Pion}.} \bibinfo{year}{2009}\natexlab{}.
\newblock \showarticletitle{{CGAL: The Computational Geometry Algorithms
  Library}}. In \bibinfo{booktitle}{{\em Proc. 17th ACM SIGSPATIAL Int. Conf.
  on Advances in Geographic Information Systems}}. \bibinfo{pages}{538{--}539}.
\newblock
\showISBNx{978-1-60558-649-6}
\showDOI{%
\url{https://doi.org/10.1145/1653771.1653865}}


\bibitem[\protect\citeauthoryear{Garey and Johnson}{Garey and Johnson}{1983}]%
        {GarJoh-SJADM-83}
\bibfield{author}{\bibinfo{person}{M.~R. Garey} {and} \bibinfo{person}{D.~S.
  Johnson}.} \bibinfo{year}{1983}\natexlab{}.
\newblock \showarticletitle{{Crossing number is NP-complete}}.
\newblock \bibinfo{journal}{{\em SIAM Journal on Algebraic Discrete Methods\/}}
  \bibinfo{volume}{4}, \bibinfo{number}{3} (\bibinfo{year}{1983}),
  \bibinfo{pages}{312{--}316}.
\newblock
\showDOI{%
\url{https://doi.org/10.1137/0604033}}


\bibitem[\protect\citeauthoryear{Goodrich}{Goodrich}{1995}]%
        {Goo-JCSS-95}
\bibfield{author}{\bibinfo{person}{Michael~T. Goodrich}.}
  \bibinfo{year}{1995}\natexlab{}.
\newblock \showarticletitle{{Planar separators and parallel polygon
  triangulation}}.
\newblock \bibinfo{journal}{{\em J. Comput. System Sci.\/}}
  \bibinfo{volume}{51}, \bibinfo{number}{3} (\bibinfo{year}{1995}),
  \bibinfo{pages}{374{--}389}.
\newblock
\showDOI{%
\url{https://doi.org/10.1006/jcss.1995.1076}}


\bibitem[\protect\citeauthoryear{Grigoriev and Bodlaender}{Grigoriev and
  Bodlaender}{2007}]%
        {GriBod-Algorithmica-07}
\bibfield{author}{\bibinfo{person}{Alexander Grigoriev} {and}
  \bibinfo{person}{Hans~L. Bodlaender}.} \bibinfo{year}{2007}\natexlab{}.
\newblock \showarticletitle{{Algorithms for graphs embeddable with few
  crossings per edge}}.
\newblock \bibinfo{journal}{{\em Algorithmica\/}} \bibinfo{volume}{49},
  \bibinfo{number}{1} (\bibinfo{year}{2007}), \bibinfo{pages}{1{--}11}.
\newblock
\showDOI{%
\url{https://doi.org/10.1007/s00453-007-0010-x}}


\bibitem[\protect\citeauthoryear{Hagberg, Swart, and Schult}{Hagberg
  et~al\mbox{.}}{2008}]%
        {NetworkX}
\bibfield{author}{\bibinfo{person}{Aric Hagberg}, \bibinfo{person}{Pieter
  Swart}, {and} \bibinfo{person}{Daniel Schult}.}
  \bibinfo{year}{2008}\natexlab{}.
\newblock \showarticletitle{{Exploring network structure, dynamics, and
  function using NetworkX}}. In \bibinfo{booktitle}{{\em Proc. 7th Python in
  Science Conf.}} \bibinfo{pages}{11{--}15}.
\newblock


\bibitem[\protect\citeauthoryear{Karduni, Kermanshah, and Derrible}{Karduni
  et~al\mbox{.}}{2016}]%
        {urban}
\bibfield{author}{\bibinfo{person}{Alireza Karduni},
  \bibinfo{person}{Amirhassan Kermanshah}, {and} \bibinfo{person}{Sybil
  Derrible}.} \bibinfo{year}{2016}\natexlab{}.
\newblock \showarticletitle{A protocol to convert spatial polyline data to
  network formats and applications to world urban road networks}.
\newblock \bibinfo{journal}{{\em Scientific data\/}}  \bibinfo{volume}{3}
  (\bibinfo{year}{2016}).
\newblock


\bibitem[\protect\citeauthoryear{Kawarabayashi}{Kawarabayashi}{2009}]%
        {Kaw-FOCS-09}
\bibfield{author}{\bibinfo{person}{Ken{-}ichi Kawarabayashi}.}
  \bibinfo{year}{2009}\natexlab{}.
\newblock \showarticletitle{{Planarity allowing few error vertices in linear
  time}}. In \bibinfo{booktitle}{{\em Proc. 50th IEEE Symp. on Foundations of
  Computer Science}}. \bibinfo{pages}{639{--}648}.
\newblock
\showDOI{%
\url{https://doi.org/10.1109/FOCS.2009.45}}


\bibitem[\protect\citeauthoryear{Klein, Rao, Rauch, and Subramanian}{Klein
  et~al\mbox{.}}{1994}]%
        {KleRaoRau-STOC-94}
\bibfield{author}{\bibinfo{person}{Philip Klein}, \bibinfo{person}{Satish Rao},
  \bibinfo{person}{Monika Rauch}, {and} \bibinfo{person}{Sairam Subramanian}.}
  \bibinfo{year}{1994}\natexlab{}.
\newblock \showarticletitle{{Faster shortest-path algorithms for planar
  graphs}}. In \bibinfo{booktitle}{{\em Proc. 26th ACM Symposium on Theory of
  Computing}}. \bibinfo{pages}{27{--}37}.
\newblock
\showISBNx{0-89791-663-8}
\showDOI{%
\url{https://doi.org/10.1145/195058.195092}}


\bibitem[\protect\citeauthoryear{Leighton}{Leighton}{1983}]%
        {Lei-VLSI-83}
\bibfield{author}{\bibinfo{person}{F.~Thomson Leighton}.}
  \bibinfo{year}{1983}\natexlab{}.
\newblock \bibinfo{booktitle}{{\em {Complexity Issues in VLSI}}}.
\newblock \bibinfo{publisher}{MIT Press}, \bibinfo{address}{Cambridge, MA}.
\newblock


\bibitem[\protect\citeauthoryear{Lipton and Tarjan}{Lipton and Tarjan}{1979}]%
        {LipTar-SJAM-79}
\bibfield{author}{\bibinfo{person}{Richard~J. Lipton} {and}
  \bibinfo{person}{Robert~E. Tarjan}.} \bibinfo{year}{1979}\natexlab{}.
\newblock \showarticletitle{{A separator theorem for planar graphs}}.
\newblock \bibinfo{journal}{{\em SIAM J. Appl. Math.\/}} \bibinfo{volume}{36},
  \bibinfo{number}{2} (\bibinfo{year}{1979}), \bibinfo{pages}{177{--}189}.
\newblock
\showDOI{%
\url{https://doi.org/10.1137/0136016}}


\bibitem[\protect\citeauthoryear{Matula and Beck}{Matula and Beck}{1983}]%
        {MatBec-JACM-83}
\bibfield{author}{\bibinfo{person}{D.~W. Matula} {and} \bibinfo{person}{L.~L.
  Beck}.} \bibinfo{year}{1983}\natexlab{}.
\newblock \showarticletitle{{Smallest-last ordering and clustering and graph
  coloring algorithms}}.
\newblock \bibinfo{journal}{{\em J. ACM\/}} \bibinfo{volume}{30},
  \bibinfo{number}{3} (\bibinfo{year}{1983}), \bibinfo{pages}{417{--}427}.
\newblock
\showDOI{%
\url{https://doi.org/10.1145/2402.322385}}


\bibitem[\protect\citeauthoryear{Ne{\v{s}}et{\v{r}}il and Ossona~de
  Mendez}{Ne{\v{s}}et{\v{r}}il and Ossona~de Mendez}{2012}]%
        {NevOss-Sparsity-12}
\bibfield{author}{\bibinfo{person}{Jaroslav Ne{\v{s}}et{\v{r}}il} {and}
  \bibinfo{person}{Patrice Ossona~de Mendez}.} \bibinfo{year}{2012}\natexlab{}.
\newblock \bibinfo{booktitle}{{\em {Sparsity: Graphs, Structures, and
  Algorithms}}}. \bibinfo{series}{Algorithms and Combinatorics},
  Vol.~\bibinfo{volume}{28}.
\newblock \bibinfo{publisher}{Springer}.
\newblock
\showISBNx{978-3-642-27874-7}
\showDOI{%
\url{https://doi.org/10.1007/978-3-642-27875-4}}


\bibitem[\protect\citeauthoryear{Pach, Radoi{\v{c}}i{\'c}, Tardos, and
  T{\'o}th}{Pach et~al\mbox{.}}{2006}]%
        {PacRadTar-DCG-06}
\bibfield{author}{\bibinfo{person}{J{\'a}nos Pach},
  \bibinfo{person}{Rado{\v{s}} Radoi{\v{c}}i{\'c}}, \bibinfo{person}{G{\'a}bor
  Tardos}, {and} \bibinfo{person}{G{\'e}za T{\'o}th}.}
  \bibinfo{year}{2006}\natexlab{}.
\newblock \showarticletitle{{Improving the crossing lemma by finding more
  crossings in sparse graphs}}.
\newblock \bibinfo{journal}{{\em Discrete Comput. Geom.\/}}
  \bibinfo{volume}{36}, \bibinfo{number}{4} (\bibinfo{year}{2006}),
  \bibinfo{pages}{527{--}552}.
\newblock
\showDOI{%
\url{https://doi.org/10.1007/s00454-006-1264-9}}


\bibitem[\protect\citeauthoryear{Pach and T{\'o}th}{Pach and T{\'o}th}{1997}]%
        {PacTot-Combinatorica-97}
\bibfield{author}{\bibinfo{person}{J{\'a}nos Pach} {and}
  \bibinfo{person}{G{\'e}za T{\'o}th}.} \bibinfo{year}{1997}\natexlab{}.
\newblock \showarticletitle{{Graphs drawn with few crossings per edge}}.
\newblock \bibinfo{journal}{{\em Combinatorica\/}} \bibinfo{volume}{17},
  \bibinfo{number}{3} (\bibinfo{year}{1997}), \bibinfo{pages}{427{--}439}.
\newblock
\showDOI{%
\url{https://doi.org/10.1007/BF01215922}}


\bibitem[\protect\citeauthoryear{Ringel}{Ringel}{1965}]%
        {Rin-AMSUH-65}
\bibfield{author}{\bibinfo{person}{Gerhard Ringel}.}
  \bibinfo{year}{1965}\natexlab{}.
\newblock \showarticletitle{{Ein Sechsfarbenproblem auf der Kugel}}.
\newblock \bibinfo{journal}{{\em Abh. Math. Sem. Univ. Hamburg\/}}
  \bibinfo{volume}{29} (\bibinfo{year}{1965}), \bibinfo{pages}{107{--}117}.
\newblock
\showDOI{%
\url{https://doi.org/10.1007/BF02996313}}


\end{thebibliography}

\end{document}